\documentclass[12pt]{iopart}
\pdfoutput=1
\usepackage{iopams}

\usepackage{physics}
\usepackage{bm}
\usepackage{comment}
\usepackage{amsmath, amsfonts, amsthm, amssymb, amscd}
\usepackage[all]{xy}
\usepackage[dvipdfmx]{graphicx}
\usepackage{pdfpages}
\usepackage{here}
\usepackage{mathdots}
\usepackage{ytableau}
\usepackage{threeparttable}
\usepackage{color}
\usepackage{caption}
\theoremstyle{plain}
\newtheorem{thm}{Theorem}[section]
\newtheorem{lem}[thm]{Lemma}
\newtheorem{cor}[thm]{Corollary}

\theoremstyle{definition}
\newtheorem{dfn}[thm]{Definition}

\theoremstyle{remark}
\newtheorem{rem}[thm]{Remark}

\begin{document}

\title[]{Jordan Decomposition of Non-Hermitian Fermionic Quadratic Forms}

\author{Shunta Kitahama${}^1$, Hironobu Yoshida${}^1$, Ryo Toyota${}^2$, and Hosho Katsura${}^{1,3,4}$}
\address{\it $^1$Department of Physics, University of Tokyo, 7-3-1 Hongo, Bunkyo-Ku, Tokyo 113-0033, Japan}
\address{\it $^2$Department of Mathematics, Texas A {\rm \&} M University, College Station, TX 77843}
\address{\it $^3$Institute for Physics of Intelligence, University of Tokyo, 7-3-1 Hongo, Bunkyo-ku, Tokyo 113-0033, Japan}
\address{\it $^4$Trans-scale Quantum Science Institute, University of Tokyo, Bunkyo-ku, Tokyo 113-0033, Japan}

\begin{abstract}
We give a rigorous proof of Conjecture 3.1 by Prosen \cite{prosen2010spectral} on the nilpotent part of the Jordan decomposition of a quadratic fermionic Liouvillian. 
We also show that the number of the Jordan blocks of each size can be expressed in terms of the coefficients of a polynomial called the $q$-binomial coefficient and describe the procedure to obtain the Jordan canonical form of the nilpotent part. 
\end{abstract}

\section{Introduction}
Fermionic quadratic forms appear in various fields of physics. 
For example, the BCS mean-field Hamiltonian for superconductivity~\cite{BCS1957, altland1997nonstandard} 
can be expressed as a quadratic form of creation and annihilation operators for electrons. 
Less obvious examples include the quantum Ising and XY chains that can be mapped to free fermions via the Jordan-Wigner transformation \cite{lieb1961two, schultz1964two}. 
Furthermore, it is known that the Kitaev honeycomb model~\cite{kitaev2006} can be mapped to a model of free fermions coupled to a static gauge field by rewriting the spin operators in terms of Majorana fermions. 
In general, if the Hamiltonian is quadratic in fermion operators, each many-body eigenstate is built out of the single-particle eigenstates. 
However, the correspondence between single-particle and many-body eigenstates can be highly 
nontrivial for non-Hermitian quadratic Hamiltonians, 
which appear in various situations, such as an effective description of dissipative systems~\cite{prosen2010spectral, prosen2008third, prosen2008quantum, prosen_zunkovic, horstmann2013, guo2017, shibata2019dissipative, shibata2020ptep, ashida2020NH, vernier2020mixing, barthel2022solving, yamanaka2023}. 

One difficulty that arises when treating non-Hermitian systems is that there may be exceptional points where the Hamiltonian becomes non-diagonalizable~\cite{ashida2020NH,Chruciski_2017}.
At an exceptional point, (proper) eigenstates do not span the entire Hilbert space. 
For this reason, we need to consider generalized eigenstates. 
Let $H$ be a non-Hermitian Hamiltonian acting on a finite-dimensional Hilbert space. 
A state $\ket{\psi}$ is a generalized (right) eigenstate of rank $d$ of $H$ with eigenvalue $E$ if $\ket{\psi}\in\ker (H-E)^d\setminus\ker(H-E)^{d-1}$.
In particular, we call a generalized eigenstate of rank one a proper eigenstate.

Exceptional points are also ubiquitous in dissipative systems~\cite{hatano2019exceptional}. Here we focus on the time evolution of the density matrix $\rho$ of an open quantum system governed by the Gorini-Kossakowski-Sudarshan-Lindblad (GKSL) equation~\cite{lindblad1976,gorini1976}: 
\begin{equation}
\frac{\dd\rho }{\dd t} = \hat{\mathcal{L}}[\rho] :=
-{\rm i} [H,\rho] + \sum_{\mu} \left( 2L_\mu \rho L_\mu^\dagger - 
\{L_\mu^\dagger L_\mu,\rho\} \right),
\label{eq:lind}
\end{equation}
where $H$ is the system Hamiltonian and $L_\mu$ are Lindblad operators describing the coupling to the environment. The superoperator $\hat{\mathcal{L}}$ is called the Liouvillian and can be thought of as an effective non-Hermitian Hamiltonian acting on the Liouville-Fock space, i.e. the Fock space of density operators. The Liouvillian $\hat{\mathcal{L}}$ may exhibit exceptional points with nontrivial Jordan blocks. In Ref.~\cite{prosen2010spectral}, Prosen studied the Jordan decomposition of a quadratic many-body Liouvillian in which $H$ and $L_\mu$ are quadratic and linear in Majorana fermion operators, respectively. 
Here we briefly review his results.

Consider the GKSL equation (\ref{eq:lind}) with 
\begin{align}
H &= \sum_{j,k=1}^{2n} w_j H_{jk} w_k,
\label{eq:ham} \\
L_\mu &= \sum_{j=1}^{2n} l_{\mu,j} w_j,
\label{eq:lin}
\end{align}
where, $w_j$, $w_k$ ($j,k=1, \ldots, 2n$) are Hermitian Majorana fermion operators satisfying the anticommutation relations
\begin{align}
\{w_j,w_k\} = 2\delta_{j,k}. 
\end{align}
Using the third quantization formalism \cite{prosen2008third}, one can rewrite the Liouvillian in the following Jordan canonical form:
\begin{align}
\hat{\mathcal{L} }= -2\sum_j \qty[\beta_j \sum_{m=1}^{\ell_j}\hat{b}_{j,m}' \hat{b}_{j,m}+\sum_{m=1}^{\ell_j-1}\hat{b}_{j,m+1}' \hat{b}_{j,m}].
\label{eq:dLL}
\end{align}
Here, $j$ is the index of the Jordan block of size $\ell_j$, $\hat{b}_{j,m}$ and $\hat{b}_{j,m}'$ are ``complex fermion operators'' that satisfy
\begin{align}\label{anticom}
\{\hat{b}_{j,m},\hat{b}_{k,n}\}=\{\hat{b}'_{j,m},\hat{b}'_{k,n}\}=0,\quad \{\hat{b}_{j,m},\hat{b}'_{k,n}\}=\delta_{j,k}\delta_{m,n},
\end{align}
and $\beta_j$, so-called rapidities, are given by the eigenvalues of the matrix whose $(i,j)$ element is $-2{\rm i} H_{ij}+2\Re\sum_{\mu} l_{\mu,i}^* l_{\mu,j}$. 
For the relation between $w_j$ and the complex fermion operators $\hat{b}_{j,m}$ and $\hat{b}_{j,m}'$, we refer the reader to the original papers by Prosen \cite{prosen2010spectral,prosen2008third}. 
Note that $\hat{b}_{j,m}'$ is not necessarily a Hermitian conjugate of $\hat{b}_{j,m}$.
We emphasize that all the arguments made in this paper do not require $\hat{b}_{j,m}'$ to be the Hermitian conjugate of $\hat{b}_{j,m}$.
The first term in the bracket of equation~(\ref{eq:dLL}) is diagonal in the complex fermion basis, whereas the second term is nilpotent. 
The Liouvillian spectrum is determined only by the first term, regardless of the addition of the second term.
In the absence of nilpotent terms, each Fock state is a proper eigenstate of $\hat{\mathcal{L}}$. 
On the other hand, the presence of the nilpotent terms makes the structure of the eigenspaces nontrivial. 
Note that nilpotent terms may also appear in the Jordan canonical form of the Liouvillian when $H_{jk}$ in equation~(\ref{eq:ham}) is a non-Hermitian matrix. 
In other words, this nilpotent term is the most general form of the non-diagonalizable part of the non-Hermitian fermionic quadratic form. 

To clarify the meaning of the second term in equation~(\ref{eq:dLL}),
let us focus on the $j$-th Jordan block. 
For notational simplicity, we suppress the subscript $j$. With this notation, the anticommutation relations in equation~(\ref{anticom}) become
\begin{align}
\{\hat{b}_{m},\hat{b}_{n}\}=\{\hat{b}'_{m},\hat{b}'_{n}\}=0,\quad \{\hat{b}_{m},\hat{b}'_{n}\}=\delta_{m,n},
\label{eq:anticom2}
\end{align}
and the nilpotent shift term in equation~(\ref{eq:dLL}) can be expressed as
\begin{equation}
\hat{\mathcal{M}}_\ell:= \sum_{k=1}^{\ell-1} \hat{b}_{k+1}' \hat{b}_{k}.
\label{eq:nilpotent}
\end{equation}
Here $\ell$ is the size of the Jordan block and can be regarded as the number of sites in a fictitious system with a one-dimensional array of $\ell$ sites, labeled by $k=1,2,...,\ell$ from left to right. 
We note in passing that $\hat{\mathcal{M}}_\ell$ can be thought of as the Hamiltonian of the uni-directional Hatano-Nelson model with open boundary conditions \cite{hatano-nelson} by identifying $\hat{b}_{k}'$ and $\hat{b}_{k}$ with the standard creation and annihilation operators for fermions.

Now, let us define the Fock basis vectors by
\begin{align}
\ket{\nu_1,\ldots,\nu_{\ell}}&:={(\hat{b}_1')}^{\nu_1} \cdots {(\hat{b}_\ell')}^{\nu_\ell}\ket{{\rm right}},\\
\bra{\nu_1,\ldots,\nu_{\ell}}&:=\bra{{\rm left}}{(\hat{b}_\ell)}^{\nu_\ell} \cdots {(\hat{b}_1)}^{\nu_1},\qquad \nu_k\in\{0,1\}. 
\end{align}
Here, $\ket{{\rm right}}$ and $\bra{{\rm left}}$ are dual of each other (i.e. $\braket{{\rm left}}{{\rm right}}=1$) and satisfy $\hat{b}_k\ket{{\rm right}}=0$ and $\bra{{\rm left}}\hat{b}_k'=0$ for any $k$.
Since $\hat{\mathcal{M}}_\ell$ commutes with the particle number operator $\hat{\mathcal{N}}_\ell:=\sum_{k=1}^{\ell}\hat{b}_k'\hat{b}_k$, we can discuss each particle number sector separately.
If $\hat{\mathcal{M}}_\ell$ operates on a Fock basis vector, we get a linear combination of all the possible states realized by moving one of the particles that can move to the right neighbor, all with the same coefficient. 
Also, we define $\hat{\Omega}_{\ell}$ to be the ``weight'' operator as follows:
\begin{align}
    \hat{\Omega}_{\ell}&:=\displaystyle\sum_{k=1}^\ell k \hat{b}'_k \hat{b}_k-\frac{1}{2}\hat{\mathcal{N}}_\ell(\hat{\mathcal{N}}_\ell+1).
\end{align}
Each Fock basis vector is an eigenstate of ${\hat \Omega}_\ell$, and the corresponding eigenvalue counts the number of times one must apply $\hat{\mathcal{M}}_\ell$ to the lowest weight state $\ket{\rm min}:=\ket{1,\dots,1,0,\dots,0}$ to reach this state. 
Let us define $\mathcal{V}_{\ell,m}$ [resp. $\mathcal{V}_{\ell,m}^r$] to be the subspace spanned by the entire states of $m$ particles [resp. $m$ particles with weight $r$]: 
\begin{align}
\mathcal{V}_{\ell,m}&:=
\{\ket{\psi}:\hat{\mathcal{N}}_{\ell}\ket{\psi}=m\ket{\psi}\},\\
\mathcal{V}_{\ell,m}^r&:=
\{\ket{\psi}:\hat{\mathcal{N}}_{\ell}\ket{\psi}=m\ket{\psi}\mathrm{and\ }\hat{\Omega}_{\ell}\ket{\psi}=r\ket{\psi}\}.
\end{align}
Note that $\mathcal{V}_{\ell,m}=\bigoplus_{r=0}^{m(\ell-m)} \mathcal{V}_{\ell,m}^r$. For convenience, let $\hat{\mathcal{M}}_{\ell,m}$ [resp. $\hat{\mathcal{M}}_{\ell,m}^r$] denote the restriction of $\hat{\mathcal{M}}_{\ell}$ to the subspace $\mathcal{V}_{\ell,m}$ [resp. $\mathcal{V}_{\ell,m}^r$]. 

To see that $\hat{\mathcal{M}}_{\ell,m}$ has a nontrivial Jordan structure, let us consider its action on the Fock basis vectors with fixed particle number $m$. Obviously, the lowest weight state is $\ket{\rm min}\in  \mathcal{V}_{\ell,m}^0$.
Operating $(\hat{\mathcal{M}}_{\ell,m})^{m(\ell-m)}$ on $\ket{\mathrm{min}}$ leads to a scalar multiple of $\ket{\mathrm{max}}:=\ket{0,\dots,0,1,\dots,1} \in  \mathcal{V}_{\ell,m}^{m(\ell-m)}$. 
When $\hat{\mathcal{M}}_{\ell,m}$ is applied to $\ket{\mathrm{max}}$ once more, it becomes 0. From this, it is easy to see that $\ket{\mathrm{max}}$ is a proper eigenstate and $\ket{\mathrm{min}}$ is a generalized eigenstate of rank $m(\ell-m)+1$. But $\ket{\mathrm{max}}$ is not the only proper eigenstate of $\hat{\mathcal{M}}_{\ell,m}$ in $\mathcal{V}_{\ell,m}$. 
Let us illustrate this with an example. Consider the case with $\ell=4$ sites and $m=2$ particles. We see that $\hat{\mathcal{M}}_{4,2}$ acts as
\begin{align}
\hat{\mathcal{M}}_{4,2}\ket{1,0,1,0} =\ket{1,0,0,1}+\ket{0,1,1,0},\\
\hat{\mathcal{M}}_{4,2}\ket{1,0,0,1} =\hat{\mathcal{M}}_{4,2}\ket{0,1,1,0}=\ket{0,1,0,1}.
\end{align}
Therefore, by defining $\ket{\pm}:=\ket{1,0,0,1}\pm\ket{0,1,1,0}$, we find that $\ket{+}$ is a generalized eigenstate of rank three, while $\ket{-}$ is a proper eigenstate.
In the same way, there are generally other proper eigenstates than $\ket{\mathrm{max}}$. 
Prosen's conjecture in~\cite{prosen2010spectral} predicts that the total number of proper eigenstates of $\hat{\mathcal{M}}_{\ell,m}$ is equal to the dimension of the subspace $\mathcal{V}_{\ell,m}^r$ with $r=\lfloor m (\ell-m)/2 \rfloor$. 
In this paper, we prove this claim rigorously. 
Furthermore, we extend this result to generalized eigenstates of $\hat{\mathcal{M}}_{\ell,m}$ by proving that $\dim\ker (\hat{\mathcal{M}}_{\ell,m})^d-\dim\ker (\hat{\mathcal{M}}_{\ell,m})^{d-1}$ is the $d$-th largest $\dim\mathcal{V}_{\ell,m}^r$ (counted with multiplicity).

This paper is organized as follows.
In Section \ref{sec:main theorem}, we show that $\hat{\mathcal{M}}_\ell$ can be regarded as a generator of $\mathfrak{sl}_2$ algebra. 
Then, using the representation theory of $\mathfrak{sl}_2$, we prove the conjecture by Prosen (Conjecture 3.1 in \cite{prosen2010spectral}). We also show that the dimensions of $\ker (\hat{\mathcal{M}}_{\ell,m})^d$ can be expressed in terms of the dimensions of subspaces $\mathcal{V}_{\ell,m}^r$. 
In Section \ref{sec:q-binomial}, we show that $\dim \mathcal{V}_{\ell,m}^r$ can be obtained as coefficients of a polynomial called the $q$-binomial coefficient, which allows us to determine the number and sizes of the Jordan blocks of $\hat{\mathcal{M}}_{\ell,m}$ explicitly. We further describe the procedure to obtain the Jordan canonical form of $\hat{\mathcal{M}}_{\ell,m}$ for general $\ell$ and $m$. 

\section{Main Theorems}
\label{sec:main theorem}
In this section, we prove Theorem \ref{prosen_conjecture} conjectured by Prosen (Conjecture 3.1 in \cite{prosen2010spectral}).
After that, we prove Theorem \ref{gen_dim} about the dimensions of
$\ker(\hat{\mathcal{M}}_{\ell,m})^d$.

\begin{thm}\label{prosen_conjecture}
The map $\hat{\mathcal{M}}_{\ell,m}^r: \mathcal{V}_{\ell,m}^r \to \mathcal{V}_{\ell,m}^{r+1}$ is injective for $r\leq\left\lfloor\frac{m(\ell-m)-1}{2}\right\rfloor$, and surjective for $r\geq\left\lfloor\frac{m(\ell-m)}{2}\right\rfloor$.
\end{thm}
 
To prove this theorem, we first show in Lemma \ref{lem:1} that $\hat{\mathcal{M}}_\ell$ can be regarded as one of the generators of the Lie algebra $\mathfrak{sl}_2$. 
Then we prove Theorem \ref{prosen_conjecture} by using well-known facts about the representation theory of $\mathfrak{sl}_2$ (See Lemmas \ref{lem:2} and \ref{lem:3}).

\begin{lem}\label{lem:1}
Let $\displaystyle\hat{\mathcal{M}}'_\ell=\sum_{k=1}^{\ell-1} k(\ell-k)\ \hat{b}_k'\hat{b}_{k+1}$ and $\hat{\mathcal{Z}}_\ell=\displaystyle\sum_{k=1}^\ell (2k-\ell-1)\hat{b}_k'\hat{b}_k$. 
Then, ${\{\hat{\mathcal{M}}_\ell,\hat{\mathcal{Z}}_\ell, \hat{\mathcal{M}}_\ell'\}}$ is an $\mathfrak{sl}_2$-triple which satisfies the following commutation relations:
\begin{align}
[\hat{\mathcal{Z}}_\ell,\hat{\mathcal{M}}_\ell]=2\hat{\mathcal{M}}_\ell,\quad[\hat{\mathcal{Z}}_\ell,\hat{\mathcal{M}}'_\ell]=-2\hat{\mathcal{M}}'_\ell,\quad
[\hat{\mathcal{M}}_\ell,\hat{\mathcal{M}}'_\ell]=\hat{\mathcal{Z}}_\ell.
\label{sl2triple}
\end{align}
\end{lem}

\begin{proof}
The last relation in equation~(\ref{sl2triple}) can be shown by a straightforward calculation:
\begin{align}
[\hat{\mathcal{M}}_\ell,\hat{\mathcal{M}}'_\ell]&=\sum_{k=1}^{\ell-1}
\sum_{k'=1}^{\ell-1} k'(\ell-k')[\hat{b}_{k+1}'\hat{b}_{k}, 
\hat{b}_{k'}'\hat{b}_{k'+1}] \nonumber \\
&=\sum_{k=1}^{\ell-1}
\sum_{k'=1}^{\ell-1} k'(\ell-k')(\delta_{k,k'}\hat{b}_{k+1}'\hat{b}_{k'+1}-\delta_{k,k'}\hat{b}_{k'}'\hat{b}_{k}) \nonumber \\
&=\sum_{k=1}^{\ell}
(2k-\ell-1)\hat{b}_k'\hat{b}_k=\hat{\mathcal{Z}}_\ell,
\end{align}
where we have used equation~(\ref{eq:anticom2}). The other two commutation relations can be shown similarly.
\end{proof}

A few remarks are in order. 
\begin{rem}
    Lemma \ref{lem:1} can be extended to the case with spatially varying couplings. In this case, the nilpotent shift operator is given by 
    \begin{equation}
        \hat{\mathcal{M}}_\ell(\{c_k\})= \sum_{k=1}^{\ell-1} c_k\ \hat{b}_{k+1}' \hat{b}_{k},
        \label{eq:difweight}
    \end{equation}
    where $c_k\in\mathbb{C}\setminus\{0\}$ can even be random. For this operator, the definition of $\hat{\mathcal{M}}'_\ell$ in Lemma \ref{lem:1} is modified to
    \begin{equation}
    \displaystyle\hat{\mathcal{M}}'_\ell(\{c_k\})=\sum_{k=1}^{\ell-1} \frac{k(\ell-k)}{c_k}\ \hat{b}_k'\hat{b}_{k+1}
    \end{equation}
    so that $\hat{\mathcal{M}}_\ell(\{c_k\})$, $\hat{\mathcal{Z}}_\ell$, and $\hat{\mathcal{M}}'_\ell(\{c_k\})$ satisfy the commutation relations of the $\mathfrak{sl}_2$-triple.   
    With this modification, the proof of Theorem \ref{prosen_conjecture} based on the representation theory of $\mathfrak{sl}_2$ also applies to the non-Hermitian Hamiltonian (\ref{eq:difweight}).
\end{rem}

\begin{rem}
A similar argument can be applied to construct an $\mathfrak{sl}_2$-triple and determine the Jordan canonical form for the bosonic uni-directional Hatano-Nelson model with open boundary conditions. This will be discussed elsewhere~\cite{kitahamaprep}.
\end{rem}

For convenience, we write the restrictions of $\hat{\mathcal{M}}'_{\ell}$, $\hat{\mathcal{Z}}_{\ell}$ and $\hat{\Omega}_{\ell}$ to $\mathcal{V}_{\ell,m}$ as $\hat{\mathcal{M}}'_{\ell,m}$, $\hat{\mathcal{Z}}_{\ell,m}$ and $\hat{\Omega}_{\ell,m}$, respectively. From Lemma \ref{lem:1}, the triple ${\{\hat{\mathcal{M}}_{\ell,m}, \hat{\mathcal{Z}}_{\ell,m},\hat{\mathcal{M}}_{\ell,m}'\}}$ is a set of generators of $\mathfrak{sl}_2$ on $\mathcal{V}_{\ell,m}$. Then, it is known that any finite-dimensional representation of $\mathfrak{sl}_2$ can always be decomposed into a direct sum of its irreducible representations (Section 6.3 of \cite{humphreys2012introduction}).

\begin{lem}\label{lem:2}
    A finite-dimensional representation of $\mathfrak{sl}_2$ is completely reducible.
\end{lem}

Furthermore, an $n$-dimensional irreducible representation of $\mathfrak{sl}_2$ is unique up to isomorphism for each dimension $n$, and satisfies the following properties (Section 6.1 of \cite{kosmann2010groups}).

\begin{lem}\label{lem:3}
    Let $\{e,h,f\}$ be an $\mathfrak{sl}_2$-triple satisfying
    \begin{align}
        [h,e]=2e, \quad [h,f]=-2f,\quad [e,f]=h,
    \end{align}
    and $(\rho,V)$ be a $2n+1$ dimensional irreducible representation of $\mathfrak{sl}_2$ for $n \in \frac{1}{2}\mathbb{Z}$.
    Then, the highest weight of $(\rho,V)$ is $2n$, i.e.
    \begin{align}
        2n=\max\{\lambda : \lambda\text{ is an eigenvalue of } \rho(h)\}
    \end{align}
    and there exists a basis $\{v_{-n},v_{-n+1},\ldots ,v_{n}\}$ of $V$ such that
    \begin{align}
        \rho(h)v_k=2kv_k,\quad\rho(f)v_k=v_{k-1},\quad \rho(e)v_k=a_kv_{k+1},
    \end{align}
    where $a_k=n(n+1)-k(k+1)$ for $k=-n,-n+1,\ldots, n$ and $v_{-n-1}=v_{n+1}=0$. 
\end{lem}

Next, let us look at the relation between $\hat{\mathcal{Z}}_\ell$ and the weight operator $\hat{\Omega}_{\ell}$.
The operator $\hat{\mathcal{Z}}_\ell$ is expressed in terms of $\hat{\Omega}_\ell$ and $\hat{\mathcal{N}}_\ell$ as
\begin{align}\label{eq:omega-z}
   \hat{\mathcal{Z}}_{\ell}=2\sum_{k=1}^\ell k\ \hat{b}_k'\hat{b}_k-(\ell+1)\hat{\mathcal{N}}_\ell=2\hat{\Omega}_{\ell}-\hat{\mathcal{N}}_\ell(\ell-\hat{\mathcal{N}}_\ell).
\end{align}
Thus,
\begin{align}
\mathcal{V}_{\ell,m}^r&=\{|\psi\rangle:\hat{\mathcal{N}}_{\ell}\ket{\psi}=m\ket{\psi}\mathrm{and\ }\hat{\Omega}_{\ell}\ket{\psi}=r\ket{\psi}\}\nonumber\\
&=\left\{|\phi\rangle:\hat{\mathcal{N}}_{\ell}\ket{\phi}=m\ket{\phi}\mathrm{and\ }\hat{\mathcal{Z}}_{\ell}\ket{\phi}=\left[2r-m(\ell-m)\right]\ket{\phi}\right\}.
\end{align}
This means that $\hat{\Omega}_{\ell}$ defined as ``weight'' actually corresponds to the weight in the sense of representation theory of Lie algebra.

Now, let us prove Theorem \ref{prosen_conjecture} using Lemmas \ref{lem:1}, \ref{lem:2}, and \ref{lem:3}.

\begin{proof}[Proof of Theorem \ref{prosen_conjecture}]
From Lemma \ref{lem:1}, we have a representation $\rho$ of $\mathfrak{sl}_2$ on $\mathcal{V}_{\ell,m}=\bigoplus_{r=0}^{m(\ell-m)} \mathcal{V}_{\ell,m}^r$ satisfying
    \begin{align}
        \rho(e)={\hat{\mathcal{M}}}_{\ell,m}, \text{ } \rho(f)={\hat{\mathcal{M}}}'_{\ell,m}, \text{ }\rho(h)=\hat{\mathcal{Z}}_{\ell,m}.
    \end{align}
    Then by Lemma \ref{lem:3}, we can decompose $(\rho,\mathcal{V}_{\ell,m})$ into a direct sum of $N$ finite-dimensional irreducible representations $\{(\rho_j,V_j)\}_{j=1}^N$ of $\mathfrak{sl}_2$ for some positive integer $N$. Here we assume that $\dim V_1\geq \dim V_2\geq\cdots\geq \dim V_N>0$. 
     Then, all the eigenvalues of $\hat{\mathcal{Z}}_{\ell,m}$ are even [resp. odd] if $m(\ell-m)$ is even [resp. odd], because $\hat{\Omega}_{\ell,m}=\frac{1}{2}\qty[\hat{\mathcal{Z}}_{\ell,m}+m(\ell-m)]$ has only integer eigenvalues.
    Thus, if $m(\ell-m)$ is even [resp. odd], $\dim V_j$ is odd [resp. even] for all $j=1,2,\ldots, N$.
    In particular, 
    \begin{equation}
    \text{$\dim V_j\geq1$ [resp. $\dim V_j\geq2$] when $m(\ell-m)$ is even [resp. odd].}\label{eq:conddim}
    \end{equation}
    We consider the following decomposition of $\mathcal{V}_{\ell,m}^r$:
    \begin{equation}
        \mathcal{V}_{\ell,m}^r= \bigoplus_{j=1}^{N}(W_j)_{\ell,m}^r \text{ where } (W_j)_{\ell,m}^r=\mathcal{V}_{\ell,m}^r \cap V_j.
    \end{equation}
Then, by Lemma \ref{lem:3},
\begin{equation}
    \dim(W_j)_{\ell,m}^r=
    \begin{cases}
    1& \text{ if } -\dim V_j+1\leq2r-m(\ell-m)\leq\dim V_j-1,\\
    0 & \text{ otherwise},
    \end{cases}
\end{equation}
and we can obtain the kernel and the image of the restriction of $\hat{\mathcal{M}}_{\ell,m}^r$ to subspace $(W_j)_{\ell,m}^r$:
    \begin{align}
    \ker \hat{\mathcal{M}}_{\ell,m}^r|_{(W_j)_{\ell,m}^r}=\ker \rho_j(e)|_{(W_j)_{\ell,m}^r}=
    \begin{cases}
    (W_j)_{\ell,m}^{r}& \!\text{ if } 2r-m(\ell-m)=\dim V_j-1,\\
    \{0\} & \!\text{ otherwise},
    \end{cases}
    \label{eq:ker}
    \end{align}
    and 
    \begin{align}
    \Im \hat{\mathcal{M}}_{\ell,m}^r|_{(W_j)_{\ell,m}^r}=\Im \rho_j(e)|_{(W_j)_{\ell,m}^r}=
    \begin{cases}
    \{0\}& \! \text{ if } 2r-m(\ell-m)=-\dim V_j,\\
    (W_j)_{\ell,m}^{r+1} & \! \text{ otherwise}.
    \end{cases}
    \label{eq:im}
    \end{align}
    If $r\leq\left\lfloor\frac{m(\ell-m)-1}{2}\right\rfloor$, the first case in equation (\ref{eq:ker}) does not occur, because the condition can be expressed as
    \begin{equation}
        \dim V_j=2r-m(\ell-m)+1\leq2\left\lfloor\frac{m(\ell-m)-1}{2}\right\rfloor-m(\ell-m)+1\leq0,
    \end{equation}
    which is incompatible with equation (\ref{eq:conddim}). Thus,
    \begin{align}   
    \ker\hat{\mathcal{M}}_{\ell,m}^r=\bigoplus_{j=1}^N \ker\hat{\mathcal{M}}_{\ell,m}^r|_{(W_j)_{\ell,m}^r}=\bigoplus_{j=1}^N\{0\}=\{0\},
    \end{align}
    and therefore $\hat{\mathcal{M}}_{\ell,m}^r$ is injective.
    If $r\geq\left\lfloor\frac{m(\ell-m)}{2}\right\rfloor$, the first case in equation (\ref{eq:im}) does not occur, because the condition can be expressed as
    \begin{equation}
        \dim V_j=-2r+m(\ell-m)\leq-2\left\lfloor\frac{m(\ell-m)}{2}\right\rfloor+m(\ell-m)=
        \begin{cases}
    0& \text{$m(\ell-m)$: even,}\\
    1 & \text{$m(\ell-m)$: odd.}
    \end{cases}
    \end{equation}
    which is incompatible with equation (\ref{eq:conddim}). Thus,
    \begin{align}   
    \Im\hat{\mathcal{M}}_{\ell,m}^r=\bigoplus_{j=1}^N \Im\hat{\mathcal{M}}_{\ell,m}^r|_{(W_j)_{\ell,m}^r}=\bigoplus_{j=1}^N(W_j)_{\ell,m}^{r+1}=\mathcal{V}_{\ell,m}^{r+1},
    \end{align}
    and therefore ${\hat{\mathcal{M}}}_{\ell,m}^r$ is subjective. 
\end{proof}

We now discuss the relation between $\dim \ker (\hat{\mathcal{M}}_{l,m})^d$ and $\dim\mathcal{V}^r_{\ell,m}$, which is crucial for obtaining the Jordan canonical form of the nilpotent shift operator in the next section. Let us prove the following Theorem \ref{gen_dim} using Lemmas \ref{lem:1}, \ref{lem:2}, and \ref{lem:3}. 
\begin{thm}\label{gen_dim}
    \begin{equation}
        \dim\ker (\hat{\mathcal{M}}_{\ell,m})^d-\dim\ker (\hat{\mathcal{M}}_{\ell,m})^{d-1}=\dim\mathcal{V}_{\ell,m}^{\lfloor \frac{m(\ell-m)-(d-1)}{2}\rfloor}.
        \label{eq:thm2}
    \end{equation}
\end{thm}
\begin{proof}[Proof of Theorem \ref{gen_dim}]
Here we use the notation used in the proof of Theorem \ref{prosen_conjecture}.
    Let $N_d$ denote the number of irreducible representations whose dimensions are larger than or equal to $d$.  In the following, we first concentrate on the case $m(\ell-m)$ even.
    In this case, the eigenvalues of $\rho_j(h)$ are even and $\dim V_j$ is odd for all $j=1,2,\ldots, N$.  When $d$ is odd, from Lemma \ref{lem:3}, we have
    \begin{align}
        \dim V_j \geq d \iff \pm(d-1) \text{ are eigenvalues of }\rho_j(h).
        \label{eq:dim_eig}
    \end{align}    
    This means that $N_d$ is the dimension of the eigenspace of $\rho(h)$ with eigenvalue $-(d-1)$. Thus,
    \begin{align}
        N_d&=\dim \ker \qty[\bigoplus_{j=1}^N\rho_j(h)+(d-1)]\nonumber\\
        &=\dim \ker \qty[\hat{\mathcal{Z}}_{\ell,m}+(d-1)]\nonumber\\
        & = \dim  \ker \qty[2\hat{\Omega}_{\ell,m}-m(\ell-m)+(d-1)] \nonumber\\
        & =\dim\mathcal{V}_{\ell,m}^{ \frac{m(\ell-m)-(d-1)}{2}}.
        \label{eq:ndodd}
    \end{align}
    When $d$ is even, since $\dim V_j$ is odd for all $j$,
    \begin{align}
    N_d=N_{d+1}=\dim\mathcal{V}_{\ell,m}^{ \frac{m(\ell-m)-d}{2}}.
    \label{eq:ndeven}
    \end{align}
    Similarly, when $m(\ell-m)$ is odd and $d$ is even [resp. odd], we have equation~\eqref{eq:ndodd} [resp. equation~\eqref{eq:ndeven}].

    Then, from Lemma \ref{lem:3},
    \begin{align}
        \dim\ker(\rho_j(e))^d = 
        \begin{cases}
        d & \qq{if}\dim V_j\geq d \\
        \dim V_j & \qq{if}\dim V_j\le d
        \end{cases}
    \end{align}
    for $d\in\mathbb{Z}_{\geq1}$, and therefore
    \begin{align}
        & \dim \ker ({\hat{\mathcal{M}}}_{\ell,m})^d- \dim \ker({\hat{\mathcal{M}}}_{\ell,m})^{d-1} \nonumber\\
        = &\dim \ker (\rho(e))^d-\dim \ker (\rho(e))^{d-1}\nonumber\\
        = &\sum_{j=1}^N\qty[\dim \ker (\rho_j(e))^d-\dim \ker (\rho_j(e))^{d-1}]\nonumber\\
        = &\sum_{j:\dim V_j\geq d}\qty[d-(d-1)]+\sum_{j:\dim V_j< d}\qty[\dim V_j-\dim V_j]\nonumber\\
        = &N_d
        \label{eq:nd}
    \end{align}
    for $d\in\mathbb{Z}_{\geq1}$. 
    Finally, since the relation between $N_d$ and $\mathcal{V}_{\ell,m}^r$ is given by \eqref{eq:ndodd} and \eqref{eq:ndeven}, we obtain equation~\eqref{eq:thm2}.
    \end{proof}

\begin{rem}
Physically speaking, the dimension of the proper eigenspace of ${\hat {\cal M}}_{\ell, m}$ is the number of irreducible representations appearing in the decomposition of the Hilbert space of $m$ spin-$\frac{\ell-1}{2}$ fermions. 
Let us explain this by an example. Consider the case with $\ell=4$, $m=2$. 
The operator $\hat{\mathcal{Z}}_4/2$ can be expressed as
\begin{align}
\frac{\hat{\mathcal{Z}}_4}{2}=\displaystyle\sum_{k=1}^4 \left(k-\frac{5}{2}\right)\hat{b}_k'\hat{b}_k
=-\frac{3}{2}\hat{b}_1'\hat{b}_1-\frac{1}{2}\hat{b}_2'\hat{b}_2+\frac{1}{2}\hat{b}_3'\hat{b}_3+\frac{3}{2}\hat{b}_4'\hat{b}_4.
\end{align}
Focusing on the $m=2$ subspace, we find that the possible eigenvalues of $\hat{\mathcal{Z}}_4/2$ are $2$, $1$, $0$ (two-fold degeneracy), $-1$, $-2$. 
This corresponds to the fact that the Hilbert space of two spin-$\frac{3}{2}$ fermions can be decomposed into the singlet and quintet subspaces. 
More complicated examples can be found in Ref.~\cite{buettner1967}.
\end{rem}

\section{Procedure to obtain the Jordan canonical form of $\hat{\mathcal{M}}_\ell$}
\label{sec:q-binomial}
We have seen that $\dim\ker (\hat{\mathcal{M}}_{\ell,m})^d$ is expressed in terms of $\dim \mathcal {V}_{\ell,m}^r$.
The question that arises is whether $\dim \mathcal {V}_{\ell,m}^r$ has an explicit expression.
We can regard $\dim \mathcal {V}_{\ell,m}^r$ as the total number of strings that can be realized by repeating the operation of moving $1$ to the right $r$ times, starting from a string of $m$ $1$'s and $\ell-m$ $0$'s from the left.
In combinatorics, it is known that these numbers are related to the polynomial called $q$-binomial coefficients defined as follows \cite{stanley2013}.

\begin{dfn}\label{q-bin-def}
($q$-binomial coefficient)
Let $n$ and $m$ be positive integers. Let us define the following polynomials of a symbol $q$.
\begin{enumerate}
\item A $q$-bracket is defined by $[n]_q:=\dfrac{1-q^n}{1-q}$.
\item A $q$-factorial is defined by $[n]_q!:=\prod_{k=1}^{n}[k]_q$.
\item A $q$-binomial coefficient is defined by $\mqty[n\\m]_q:=\dfrac{[n]_q!}{[m]_q![n-m]_q!}$.
\end{enumerate}
\end{dfn}

The following theorem relates the dimensions of the $\mathcal{V}_{\ell,m}^r$ and $q$-binomial coefficients.
This result is known in algebraic combinatorics \cite{stanley2013}, but for the reader's convenience, we give a simple proof in Appendix.
\begin{thm}\label{q-bin-state}
\begin{align}
\label{q-bin-statement}
    \mqty[\ell\\m]_q=\sum_{r=0}^{m(\ell-m)}q^r \dim \mathcal{V}_{\ell,m}^r.
\end{align}
\end{thm}
Combining Theorems \ref{gen_dim} and \ref{q-bin-state}, we arrive at the following: 
\begin{cor}
The difference $\dim\ker (\hat{\mathcal{M}}_{\ell,m})^d-\dim\ker (\hat{\mathcal{M}}_{\ell,m})^{d-1}$ is the coefficient of $q^{\lfloor \frac{m(\ell-m)-(d-1)}{2}\rfloor}$ in $\mqty[\ell\\m]_q$.
\end{cor}

This means that $\dim\ker (\hat{\mathcal{M}}_{\ell,m})^d-\dim\ker (\hat{\mathcal{M}}_{\ell,m})^{d-1}$ is the $d$-th largest coefficient of $\mqty[\ell\\m]_q$ (counted with multiplicity). 
Let us give an example for $\ell=6$ and $m=3$.
In this case, the $q$-binomial coefficient reads
\begin{align}
    \mqty[6\\3]_q=1 + q + 2 q^2 + 3 q^3 + 3 q^4 + 3 q^5 + 3 q^6 + 2 q^7 + q^8 + q^9.
\end{align}
From this polynomial, we can obtain $\dim\ker (\hat{\mathcal{M}}_{6,3})^d-\dim\ker (\hat{\mathcal{M}}_{6,3})^{d-1}$ as shown in Table 1.
From the table, it is clear that $\hat{\mathcal{M}}_{6,3}$ is decomposed into three Jordan blocks whose sizes are 10, 6 and 4.
\begin{table}[H]
\centering
  \caption{\small Table of $\dim\ker (\hat{\mathcal{M}}_{6,3})^d-\dim\ker (\hat{\mathcal{M}}_{6,3})^{d-1}$.}
\begin{tabular}{c|cccccccccc}
\hline
    $d$ & 1&2&3&4&5&6&7&8&9&10 \\ \hline
    $\dim\ker (\hat{\mathcal{M}}_{6,3})^d-\dim\ker (\hat{\mathcal{M}}_{6,3})^{d-1}$& 3&3&3&3&2&2&1&1&1&1 \\\hline
 \end{tabular}
\end{table}

Similarly, one can determine the number and sizes of the Jordan blocks of $\hat{\mathcal{M}}_{\ell,m}$ for general $\ell$ and $m$. The procedure can be summarized in the following theorem:

\begin{cor}
    Jordan blocks of $\hat{\mathcal{M}}_{\ell,m}$ of size $d$ $(1\leq d\leq m(\ell-m)+1)$ appear only if $ m(\ell-m)-d $ is odd, and in this case,
    the number of them is the difference between the coefficients of $q^{\frac{m(\ell-m)-d+1}{2}}$ and $q^{\frac{m(\ell-m)-d-1}{2}}$ in $\mqty[\ell\\m]_q$. 
\end{cor}

\begin{proof}
    The number of Jordan blocks of $\hat{\mathcal{M}}_{\ell,m}$ of size $d$ is $N_d-N_{d+1}$. Due to Theorem \ref{gen_dim},
    \begin{align}
        N_d-N_{d+1}&=\dim\mathcal{V}_{\ell,m}^{\lfloor \frac{m(\ell-m)-(d-1)}{2}\rfloor}-\dim\mathcal{V}_{\ell,m}^{\lfloor \frac{m(\ell-m)-d}{2}\rfloor}\nonumber\\
        &=
        \begin{cases}
\dim\mathcal{V}_{\ell,m}^{ \frac{m(\ell-m)-d+1}{2}}-\dim\mathcal{V}_{\ell,m}^{ \frac{m(\ell-m)-d-1}{2}} & m(\ell-m)-d:{\rm odd}, \\
0 & m(\ell-m)-d:{\rm even}.
\end{cases}
    \end{align}
Combining this result with Theorem \ref{q-bin-state}, we obtain the desired result. 
\end{proof}

\section*{Acknowledgments}
We thank Toma$\check{\rm z}$ Prosen for encouraging us to publish this work. 
S.K. was supported by JSPS KAKENHI Grant-in-Aid for JSPS fellows Grant No. JP23KJ0738, the Forefront Physics and Mathematics Program to Drive Transformation, the University of Tokyo. 
H.Y. was supported by JSPS KAKENHI Grant-in-Aid for JSPS fellows Grant No. JP22J20888, the Forefront Physics and Mathematics Program to Drive Transformation, and JSR Fellowship, the University of Tokyo.
H.K. was supported by JSPS KAKENHI Grants No. JP18K03445, No. JP23H01093, No. 23H01086, and MEXT KAKENHI Grant-in-Aid for Transformative Research Areas A “Extreme Universe” (KAKENHI Grant No. JP21H05191).

\section*{Appendix: Proof of Theorem \ref{q-bin-state}}
In this appendix, we give a proof of Theorem \ref{q-bin-state}, although an equivalent statement is proved in \cite{stanley2013}.
A straightforward calculation shows that the $q$-binomial coefficients defined in Definition \ref{q-bin-def} satisfy the following recurrence formula:
\begin{align}
\label{recurence}
    \mqty[\ell\\m]_q=q^m\mqty[\ell-1\\m]_q+\mqty[\ell-1\\m-1]_q,
\end{align}
with the initial conditions $\mqty[\ell\\0]_q=\mqty[\ell\\l]_q=1$ for $\ell\geq0$.
Let us prove Theorem \ref{q-bin-state} using this recurrence formula.

\begin{proof}[Proof of Theorem \ref{q-bin-state}]
Let $B_{\ell,m}(q)$ denote the right-hand side of equation~(\ref{q-bin-statement}).
We shall show that $B_{\ell,m}(q)$ and $\mqty[\ell\\m]_q$ satisfy the same recurrence formula and the initial conditions.
Let $\mathcal{U}^r_{\ell,m}$ [resp. $\mathcal{W}^r_{\ell,m}$] be the subspace of $\mathcal{V}^r_{\ell,m}$ spanned by Fock basis vectors of the form $\ket{0,*,*,*,......,*}$ [resp. $\ket{1,*,*,*,......,*}$].
Since $\mathcal{V}_{\ell,m}^r$ can be decomposed into a direct sum of $\mathcal{U}^r_{\ell,m}$ and $\mathcal{W}^r_{\ell,m}$, we see that
    \begin{align}
    \label{q-binomial-dim}
        \dim\mathcal{V}_{\ell,m}^r=\dim(\mathcal{U}_{\ell,m}^r \oplus \mathcal{W}_{\ell,m}^r)=\dim \mathcal{U}_{\ell,m}^r+ \dim \mathcal{W}_{\ell,m}^r
        =\dim\mathcal{V}_{\ell-1,m}^{r-m}+\dim\mathcal{V}_{\ell-1,m-1}^{r}.
    \end{align}
Therefore, we have
\begin{align}
    q^mB_{\ell-1,m}(q)+B_{\ell-1,m-1}(q)&=\sum_{r=0}^{m(\ell-m-1)}q^{r+m} \dim \mathcal{V}_{\ell-1,m}^r+\sum_{r=0}^{(m-1)(\ell-m)}q^{r} \dim \mathcal{V}_{\ell-1,m-1}^r\nonumber\\
    &=\sum_{r'=m}^{m(\ell-m)}q^{r'} \dim \mathcal{V}_{\ell-1,m}^{r'-m}+\sum_{r=0}^{(m-1)(\ell-m)}q^{r} \dim \mathcal{V}_{\ell-1,m-1}^r\nonumber\\
    &=\sum_{r=0}^{m(\ell-m)}q^{r} (\dim \mathcal{V}_{\ell-1,m}^{r-m}+\dim \mathcal{V}_{\ell-1,m-1}^{r})\nonumber\\
    &=\sum_{r=0}^{m(\ell-m)}q^r \dim \mathcal{V}_{\ell,m}^r=B_{\ell,m}(q),
\end{align}
where we used $\dim\mathcal{V}_{\ell,m}^r=0$ if $r<0$, $r>m(\ell-m)$.
This is the same recurrence relation as equation~(\ref{recurence}).
In addition, since $B_{\ell,0}(q)=B_{\ell,\ell}(q)=1$, $B_{\ell,m}(q)$ satisfies the same initial conditions as the $q$-binomial coefficient $\mqty[\ell\\m]_q$. Thus, $B_{\ell,m}(q)=\mqty[\ell\\m]_q$.
\end{proof}

\section*{References}

\bibliography{reference}
\bibliographystyle{iopart-num}

\end{document}